\documentclass{amsart}

\usepackage{graphicx}

\newtheorem{theorem}{Theorem}[section]
\newtheorem{lemma}[theorem]{Lemma}
\newtheorem{prop}[theorem]{Proposition}

\theoremstyle{definition}
\newtheorem{definition}[theorem]{Definition}
\newtheorem{example}[theorem]{Example}

\theoremstyle{remark}

\numberwithin{equation}{section}

\begin{document}

\title[Some interpretations in QFT]{Some combinatorial interpretations in perturbative quantum field theory}

\author{Karen Yeats}
\address{Department of Mathematics, Simon Fraser University, 8888 University Dr, Burnaby BC, Canada}
\email{karen.yeats@sfu.ca}


\date{}

\begin{abstract}
This paper will describe how combinatorial interpretations can help us understand the algebraic structure of two aspects of perturbative quantum field theory, namely analytic Dyson-Schwinger equations and periods of scalar Feynman graphs.  The particular examples which will be looked at are, a better reduction to geometric series for Dyson-Schwinger equations, a subgraph which yields extra denominator reductions in scalar Feynman integrals, and an explanation of a trick of Brown and Schnetz to get one extra step in the denominator reduction of an important particular graph.
\end{abstract}

\maketitle

\section{Introduction}

In the subtle worlds of periods and quantum field theory it can often be hard to clearly see one's way.  Combinatorial interpretations can help by giving us explicit objects to get our hands dirty with.  This paper will discuss two situations with such interpretations in perturbative quantum field theory.

The first situation, which is discussed in Section \ref{sec 2}, concerns combinatorial understanding of analytic aspects of Dyson-Schwinger equations.  One example of this is the recent chord diagram expansion for a class of Dyson-Schwinger equations given by Nicolas Marie and the author in \cite{MYchord}.  For those readers who came to this proceedings volume looking for information related to the conference talk, the results of \cite{MYchord} are summarized.  Another example is the reduction process of \cite{kythesis, Ymem}.  In this paper I will prove an observation which provides a conceptual clean-up to the reduction to geometric series of \cite{kythesis, Ymem}.

The second situation, which is discussed in Section \ref{sec cov}, concerns understanding when Brown's denominator reduction approach to calculating periods of scalar Feynman graphs \cite{Brbig} yields results which are nicer than expected.  In particular we will look at when we get a ``free'' factorization, and one circumstance where we can proceed an extra step after denominator reduction fails, using a change of variables, as was done by Brown and Schnetz in section 6.2 of \cite{BrS} in order to find a K3 surface.  Note that tricks such as the Brown-Schnetz change of variables were found without such combinatorial interpretations, however the interpretation can tell us something about \emph{why} the change of variables does its job.

The new material is Theorem \ref{thm geometric}, Proposition \ref{prop new free}, and Theorem \ref{thm cov}, but also interesting is the overall story of nice interpretations involving these results woven together with results of Brown \cite{Brbig}, Brown and Schnetz \cite{BrS}, the author with Nicolas Marie \cite{MYchord}, and the author \cite{kythesis, Ymem}.

\section{Dyson-Schwinger equations and the chord diagram expansion}\label{sec 2}

\subsection{The Hopf algebra of rooted trees}

Let $\mathcal{H}$ be a renormalization Hopf algebra.  It will suffice for our purposes here to take $\mathcal{H}$ to be the Connes-Kreimer Hopf algebra of rooted trees \cite{ck0}, but this framework functions in the same way for Hopf algebras of Feynman graphs \cite{Ymem, kythesis}.  

As an algebra, the Connes-Kreimer Hopf algebra is the polynomial algebra over the set of rooted trees.  Monomials can also be viewed as disjoint unions of trees and hence are forests.  Rooted trees are taken without regards to planar embedding, and with no restrictions on valence.  The size of a tree is its number of vertices.  This yields a grading on the algebra.  

The coalgebra structure is not essential for the presentation of the main result of this section, however it underlies this entire theory and thus will be described briefly.  To define the coproduct we need a few auxiliary definitions.  An \emph{admissible cut} on a rooted tree $T$ is a set of vertices (possibly empty) of $T$ with the property that no vertex is an ancestor of another.  Given an admissible cut $c$ of a rooted tree $T$ the \emph{pruned part}, $P_c(T)$ is the forest of trees rooted at the vertices of $c$; the \emph{root part}, $R_c(T)$, is the tree resulting from removing $P_c(T)$ from $T$.  Note that admissible cuts are usually defined based on edge cuts but the definitions are equivalent provided the trivial and full cuts are included as admissible cuts in the edge-based definitions. 
   The coproduct for the Connes-Kreimer Hopf algebra is 
\[
\Delta(T) = \sum_{\substack{c \text{ admissible}\\\text{cut of }T}}P_c(T)\otimes R_c(T)
\]
on trees $T$ and extended as an algebra homomorphism.  Note that the empty cut yields the term $1\otimes T$ and the cut which consists of the root alone yields $T\otimes 1$.  The counit is given by $\eta(1)=1$ and $\eta(T)=0$ for any tree $T$ and extended as an algebra homomorphism.  This defines a graded bialgebra.  The grading gives the antipode recursively from its defining property.

\subsection{Combinatorial Dyson-Schwinger equations}

Combinatorial Dyson-Schwinger equations are recurrences, sometimes systems of recurrences, with solutions which are formal power series in $\mathcal{H}$.  Some examples in rooted trees follow  the definitions.  For Feynman graphs, the combinatorial Dyson-Schwinger equations are defined so that the solutions are sums of all graphs with a given external structure.  Thus once we apply Feynman rules we obtain the Green functions of the theory; these are physically important quantities.

\begin{definition}
In the Connes-Kreimer Hopf algebra let $B_+$ be the endomorphism which takes a forest $T_1T_2\cdots T_k$ to the rooted tree where the children of the root are $T_1,T_2,\ldots,T_k$.
\end{definition}
For example 
\[
  B_+\left(\raisebox{-1ex}{\includegraphics{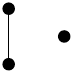}}\right) = \raisebox{-3ex}{\includegraphics{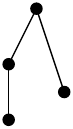}}
\]  

In other renormalization Hopf algebras there may be many $B_+$ operations indexed by primitive elements of the Hopf algebra.  The $B_+$ must be Hochschild 1-cocycles, that is they satisfy
\[
  \Delta B_+ = (\text{id}\otimes B_+)\Delta + B_+\otimes 1
\]
For the $B_+$ of rooted trees this equation can be seen to hold by considering the decomposition of a tree into the root and its subtrees.  In other renormalization Hopf algebras it may be necessary to work in a quotient Hopf algebra in order to have the 1-cocycle property; this is the situation in gauge theories \cite{vS, vS2}.

\begin{definition}
\emph{Combinatorial Dyson-Schwinger equations} are systems of equations of the form
\begin{align*}
  X_1(x) & = 1 \pm \sum_{k \geq 1}x^kB_+^{k,1}(P_{k,1}(X_1(x),\ldots,X_t(x))) \\
  \vdots & \\
  X_t(x) & = 1 \pm \sum_{k \geq 1}x^kB_+^{k,t}(P_{k,t}(X_1(x),\ldots,X_t(x)))
\end{align*}
where the $B_+^{k,t}$ are 1-cocycles and the $P_{k,t}$ are rational functions.  
\end{definition}
There are two important restrictions to make, the first of which brings us to the key cases in physics, and the second of which brings us to the problems of this section.

First, let us restrict our interest to 
\[
P_{k,t}(X_1,\ldots, X_t) = X_t Q^k
\]
where
\[
Q = \prod_{r=1}^t X_r^{s_r}
\]
with the $s_r$ integers.  The signs in the Dyson-Schwinger equation are then taken to be the signs of the corresponding $s_r$; note that this is the opposite sign convention from \cite{Ymem, kythesis}.  The series $X_r$ all begin with $1$ so taking formal inverses makes sense.  Here $Q$ is the combinatorial avatar of the physicists' \emph{invariant charge}.  The fact that there is such a $Q$ in physically important cases, such as quantum chromodynamics, comes down to the same quotient which made the $B_+$s into 1-cocycles \cite[section 2]{anatomy}.

Second, let us further restrict to single equation Dyson-Schwinger equations
\begin{equation}\label{eq sDSE}
X(x) = 1 + \text{sgn}(s) \sum_{k \geq 1}x^kB_+^{k}(X^{1+ks}(x)) 
\end{equation}
where $\text{sgn}(s)$ is the sign of $s$.

Consider a few examples of such equations to get a feel for how they work.

\begin{example}
  Consider
  \[
  X(x) = 1+xB_+(X(x))
  \]
  in the Connes-Kreimer Hopf algebra of rooted trees.  This is $s=0$ in \eqref{eq sDSE}.   Expanding we get
  \begin{align*}
  X(x) & = 1+ O(x) \\
   & = 1+ xB_+(1 + O(x)) = 1+ x\includegraphics{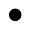} + O(x^2) \\
   & = 1 + xB_+(1+x\includegraphics{B} + O(x^2)) = 1 + x\includegraphics{B} + x^2\includegraphics{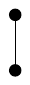} +O(x^3)\\
  & = 1 + xB_+(1 + x\includegraphics{B} + x^2\includegraphics{BB} +O(x^3)) = 1 + x\includegraphics{B} + x^2\includegraphics{BB} + x^3\includegraphics{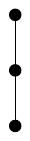} + O(x^4)
  \end{align*}
  One can quickly see that the coefficient of $x^n$ in $X(x)$ will be the tree with $n$ vertices and no branching.
\end{example}

\begin{example}
  Consider
  \[
  X(x) = 1+xB_+(X^2(x))
  \]
  in the Connes-Kreimer Hopf algebra.  This is $s=1$ in \eqref{eq sDSE}.   Expanding we get
  \begin{align*}
  X(x) & = 1+ O(x) \\
   & = 1+ xB_+(1 + O(x)) = 1 + x\includegraphics{B} + O(x^2) \\
   & = 1 + xB_+((1 + x\includegraphics{B} + O(x^2))^2) = 1 + x\includegraphics{B} + 2x^2\includegraphics{BB} + O(x^3) \\
   & = 1 + xB_+\left(\left(1 + x\includegraphics{B} + 2x^2\includegraphics{BB} + O(x^3)\right)^2\right) \\
  & = 1 + x\includegraphics{B} + 2x^2\includegraphics{BB} + 4x^3\includegraphics{BBB} + x^3\includegraphics{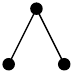} + O(x^4)
  \end{align*}
  Here we are getting all binary trees weighted by the number of distinct ways that the children of every vertex can be assigned to be left or right children, with at most one of each for each vertex.  This can be proved with a quick induction.
\end{example}

\begin{example}
  Consider
  \[
  X(x) = 1-xB_+(X^{-1}(x))
  \]
  in the Connes-Kreimer Hopf algebra.  This is $s=-2$ in \eqref{eq sDSE}.  Expanding we get
  \begin{align*}
  X(x) & = 1+ O(x) \\
   & = 1- xB_+(1 + O(x)) = 1 - x\includegraphics{B} + O(x^2) \\
   & = 1 - xB_+(1 + x\includegraphics{B} + O(x^2)) = 1 - x\includegraphics{B} - x^2\includegraphics{BB} + O(x^3) \\
   & = 1 - xB_+\left(1 + x\includegraphics{B} + x^2\includegraphics{BB} + \left( x\includegraphics{B} + x^2\includegraphics{BB}\right)^2 + O(x^3)\right) \\
  & = 1 - x\includegraphics{B} - x^2\includegraphics{BB} - x^3\includegraphics{BBB} - x^3\includegraphics{BLR} + O(x^4) \\
  & = 1 - x\includegraphics{B} - x^2\includegraphics{BB} - x^3\includegraphics{BBB} - x^3\includegraphics{BLR} - x^4\includegraphics{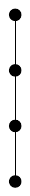} - 2x^4\includegraphics{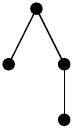} - x^4\includegraphics{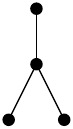} - x^4\includegraphics{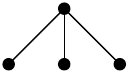} \\
  & \quad \quad + O(x^5)
  \end{align*}
  Here we are getting all rooted trees weighted by the number of distinct plane embeddings. 

  This is a very important example in what follows as it is the case for which we can solve the corresponding analytic Dyson-Schwinger equation using a chord diagram expansion.
\end{example}

\subsection{Analytic Dyson-Schwinger equations}

For a workshop on periods we cannot be satisfied with combinatorial Dyson-Schwinger equations as we need to have analytic information, thence periods.

We can no longer completely ignore the Feynman graphs as we will need the analytic structure of the primitives.  For the purposes of this paper it suffices to consider situations where we insert in only one edge $e$ of each primitive $\gamma$.  Take the Feynman integral of $\gamma$ and regularize it by raising the propagator associated to $e$ to the power $1-\rho$.  This converges for small nonzero values of $\rho$.  Fix the external momenta and expand the result as a Laurent series in $\rho$.   Call this series
\[
F_\gamma(\rho)
\]
Assume that $F_\gamma(\rho)$ has a first order pole at $0$.  This is the situation in the physical cases of interest since $\gamma$ is primitive.

\begin{definition}\label{def aDSE}
Given a combinatorial Dyson-Schwinger of the form
\[
X(x) = 1 + \text{sgn}(s) \sum_{k \geq 1}x^kB_+^{k}(X^{1+ks}(x)) 
\]
the associated \emph{analytic Dyson-Schwinger equation} is
\[
G(x,L) = 1 + \text{sgn}(s)\sum_{k \geq 1}x^kG\left(x,\frac{d}{d(-\rho)}\right)^{1+sk}(e^{-L\rho}-1)F_{k}(\rho) \bigg|_{\rho=0}
\]
where 
\[
  F_k(\rho) = \sum_{i=-1}^{\infty} f_{k,i+1}\rho^i
\]
with the $f_{k,j}$ viewed as given (by physics).
\end{definition}
Systems can be treated similarly \cite[subsection 3.3.2]{Ymem}, but it is messier and not important for the present purposes.

This definition doesn't look much like the usual form of Dyson-Schwinger equations that a physicist would be familiar with.  The best way to see the connection between this definition and the more usual form is an example

\begin{example}\label{eg yukawa}
This example is from \cite{bkerfc}.
Consider graphs built from 
\[
\includegraphics{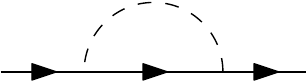}
\] 
inserted into itself in all possible ways.  This yields graphs such as
\[
\includegraphics{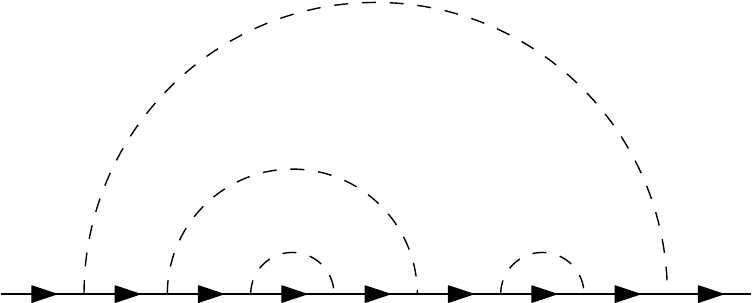}
\]
The tree structure of the insertions gives the combinatorial Dyson-Schwinger equation
\[
  X(x) = \mathbb{I} - x B_+
\left(\frac{1}{X(x)}\right)
\]

Applying Feynman rules we obtain an analytic Dyson-Schwinger equation in a form closer to what one would find in physics sources \cite{bkerfc}, specifically
\begin{equation}\label{physics dse}
G(x, L) = 1 - \frac{x}{q^2}\int d^4 k \frac{k \cdot q}{k^2 G(x,
  \log k^2)(k+q)^2} - \cdots \bigg|_{q^2 = \mu^2}
\end{equation}
where $L = \log(q^2/\mu^2)$, $q$ is the external momentum, and $\mu$ is a reference momentum.  See \cite{Ymem} Example 3.5 for further details.  Note that the structure of this equation is quite similar to the combinatorial Dyson-Schwinger equation: $B_+$ has become an integral operator and those bits of the Feynman integral given by the primitive;  the argument to $B_+$ has become the recursive appearance of $G(x,\log k^2)$.  Note also that renormalization is taken care of by a single subtraction at a fixed value of the momentum because the renormalization of the subgraphs is done recursively by the equation itself.

This does not, however, look much like the analytic Dyson-Schwinger equations defined in Definition \ref{def aDSE}. Example 3.7 of \cite{Ymem} begins with the above example, proceeds to expand $G(x,L)$ in $L$, convert logarithms to powers using $\frac{d^ky^{\rho}}{d\rho^k}  |_{\rho=0} = \log^k(y)$, swap the order of the operators, and thus obtains
\begin{equation}\label{my eg dse}
 G(x,L) = 1- x G\left(x,\frac{d}{d(-\rho)}\right)^{-1} (e^{-L\rho}-1)F(\rho) \big|_{\rho=0}
\end{equation}
where $F(\rho)$ is the Feynman integral of the primitive with the propagator we are inserting on regularized, and the integral evaluated at $q^2=1$.  The calculations take about a page but do not contain any notable subtleties.  

This example is the motivation for Definition \ref{def aDSE}.  The basic steps which converted \eqref{physics dse} to \eqref{my eg dse} could apply to any Dyson-Schwinger equation in its usual analytic form provided everything is sufficiently well-behaved analytically to swap the order of the integrals and derivatives as described.  To avoid such analytic nuisance I simply \emph{define} the analytic Dyson-Schwinger equations to be the formal outcome of this process, which is what is given by Definition \ref{def aDSE}.  In this context the expansion of $F(\rho)$ is viewed as given by physics because, as in the example, $F(\rho)$ is the integral for the primitive regularized at the insertion place and with the external momenta fixed.
\end{example}

It is tempting, when we are considering the number theory or algebraic geometry of Feynman graphs, to focus on graphs individually.  There is certainly a lot of interesting mathematics in each graph, and this will be the approach in Section \ref{sec cov}.  However, it is important that the mathematical investigation of Feynman graphs does not begin and end at the individual graph level.  The solutions to Dyson-Schwinger equations are physically meaningful quantities while individual graphs are largely not.  Thus the sums of graphs from Dyson-Schwinger equations yield the periods of greatest interest.

\subsection{Chord diagrams}

The main result of \cite{MYchord}, which was the topic of the author's talk at the conference for which these are the proceedings, is a solution to the particular Dyson-Schwinger equation
\begin{equation}\label{eq chord dse}
G(x,L) = 1 - xG\left(x,\frac{d}{d(-\rho)}\right)^{-1}(e^{-L\rho}-1)F(\rho) \bigg|_{\rho=0}
\end{equation}
as an expansion in $x$ and $L$ indexed by chord diagrams and with coefficients monomials in the $f_i$ where
\[
  F(\rho) = \sum_{i=-1}^{\infty} f_{i+1}\rho^i
\]
\cite{MYchord} is joint work with Nicolas Marie.

The objects we need are rooted connected chord diagrams.
\begin{definition}
  A \emph{perfect matching} of a finite set $S$ is a set of pairs of elements of $S$ such that every element of $S$ is in exactly one pair.
\end{definition}
\begin{definition}
  A \emph{rooted chord diagram} with $n$ chords is a perfect matching of $\{1,2,\ldots,2n\}$.  The pairs of the perfect matching are called \emph{chords} and are ordered by the order of their smaller elements.  The \emph{root chord} is the pair including $1$.
\end{definition}
To visualize a rooted chord diagram put the points $1,\ldots, 2n$ in counterclockwise order around a circle, let $1$ be the root vertex, and draw chords through the circle joining the pairs of points of the matching.  For example
\[
\includegraphics{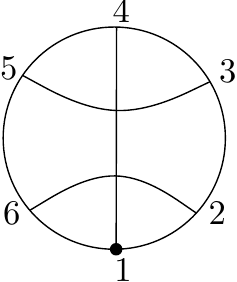}
\]
In the illustrations the root vertex will be marked with a dark dot.

The order of the chords is counterclockwise by their first appearance starting from the root. 

\begin{definition}
  The \emph{intersection graph} of a rooted chord diagram $C$ is the graph with one vertex for each chord of $C$ and an edge between vertices iff the chords cross, that is, iff the chords are pairs $\{a,b\}$, $\{c,d\}$ with $a<c<b<d$.  

The \emph{oriented intersection graph} of a rooted chord diagram $C$ is the intersection graph of $C$ with all edges oriented from the smaller chord to the larger chord.
\end{definition}

For example the oriented intersection graph of the previous chord diagram is
\[
\includegraphics{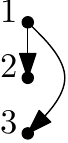}
\]

\begin{definition}
  A rooted chord diagram $C$ is \emph{connected} if the intersection graph is connected.
\end{definition}

For example the previous chord diagram is connected while
\[
\includegraphics{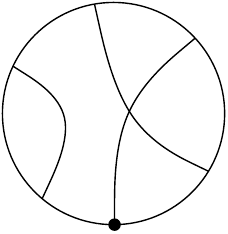}
\]
is not connected.

Next we need some more specialized definitions for our situation.
\begin{definition}
  A chord of a rooted connected chord diagram $C$ is called \emph{terminal} if the corresponding vertex of the oriented intersection graph of $C$ has no outgoing edges.
\end{definition}

 There is another order which is important in the following.  It is called the intersection order, and is defined recursively.

\begin{definition}
Let $C$ be a rooted connected chord diagram.  The \emph{intersection order} of the chords of $C$ is defined by the following procedure.  
\begin{itemize}
  \item The root chord of $C$ is the first chord in the intersection order.  
  \item Remove the root chord of $C$ and let $C_1$, \ldots $C_k$ be the connected components of the remaining chord diagram, ordered by the counterclockwise order of their first vertex.
  \item Order the chords of each of $C_1$, \ldots $C_k$ inductively in the intersection order.  Order the chords of $C$ with the root chord first followed by all chords of $C_1$ in intersection order, then all chords of $C_2$ in intersection order, and so on.
\end{itemize}
\end{definition}

For example the chords of the following chord diagrams are labelled in intersection order
\[
\includegraphics{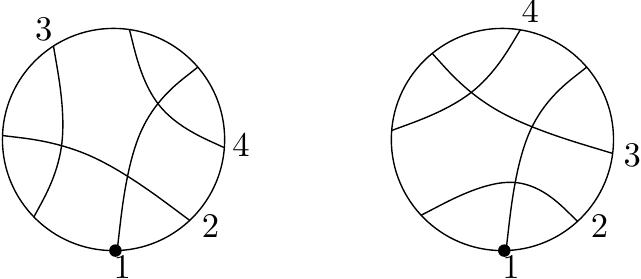}
\]
Note that the intersection order may or may not correspond to the counterclockwise order.

\begin{definition}
  Let $C$ be a rooted connected chord diagram with $n$ chords and let $f_0, f_1,\ldots $ be a countable set of indeterminates
\begin{enumerate}
  \item Let $T(C)=(i_1< i_2 < \cdots < i_k)$ be the list in increasing order of indices of terminal chords of $C$ in the intersection order.
  \item Let $b(C)$ be the first element of $T(C)$
  \item Let $\delta(C) = (\underbrace{0,\ldots,0}_{n-k \text{ times}}, i_2-i_1,i_3-i_2,\ldots,i_k-i_{k-1})$ be the list of differences of successive elements of $T(C)$ padded with $0$s so that $|\delta(C)| = n-1$.
  \item Let $f_C = \prod_{i\in \delta(C)}f_i$
\end{enumerate}
\end{definition}

The $f_C$ are what is needed to build up the monomials of the chord diagram expansion of the solution of \eqref{eq chord dse}.

The main result of \cite{MYchord} with Nicolas Marie is
\begin{theorem}
\[
G(x,L) = 1 - \sum_{i\geq 1}\frac{(-L)^i}{i!}\sum_{\substack{C \\ b(C) \geq i}}x^{|C|}f_{C}
f_{b(C)-i} 
\]
where the sum is over rooted connected chord diagrams with the indicated restriction,
solves the Dyson-Schwinger equation
\[
G(x,L) = 1 - xG\left(x,\frac{d}{d(-\rho)}\right)^{-1}(e^{-L\rho}-1)F(\rho) \big|_{\rho=0}
\]
where
\[
F(\rho)  = \frac{f_{0}}{\rho} + f_1 + f_2\rho + f_3\rho^2 + \cdots
\]
%  G(x,L) & = 1 - \sum_{n \geq 1} \gamma_n(x)L^n
\end{theorem}

The proof of the result is the body of \cite{MYchord}.  It involves two further recurrences, one of which generalizes a classical chord diagram recurrence of Stein \cite{NWchord} and one of which involves going through a rooted tree construction.  This result is interesting because it gives the Green function $G(x,L)$ as a sort of multivariate generating function for chord diagrams.  Expand
\[
G(x,L) = 1 - \sum_{n \geq 1} \gamma_n(x)L^n
\]
In this situation the analogue of the beta function of the theory is simply
\[
  \beta(x) = -2x\gamma_1(x)
\]
and hence is also given by an expansion over chord diagrams.  More physically realistic cases will be given by systems of Dyson-Schwinger equations and then the beta function will be a linear combination of the $\gamma_1$s for each Green function.  There are some hints that similar combinatorial expansions will hold for other single equation Dyson-Schwinger equations and then ultimately for systems, but as of yet there are no further results along those lines.

\subsection{Reduction to geometric series}

Let us now return to the more general form of the Dyson-Schwinger equation defined in Definition \ref{def aDSE}
\begin{equation}\label{eq aDSE redux}
G(x,L) = 1 + \text{sgn}(s)\sum_{k \geq 1}x^kG\left(x,\frac{d}{d(-\rho)}\right)^{1+sk}(e^{-L\rho}-1)F_{k}(\rho) \bigg|_{\rho=0}
\end{equation}
where 
\[
  F_k(\rho) = \sum_{i=-1}^{\infty} f_{k,i+1}\rho^i
\]
This form of Dyson-Schwinger equation was chosen for its amenability to algebraic and combinatorial analysis.  

We will need an important result from \cite{Ymem, kythesis}.  Writing the renormalization group equation in this language gives (see \cite[Section 4.1]{Ymem})
\begin{prop}\label{prop gammak rec}
  With $G(x,L)$ as above, and writing 
  \[
  G(x,L) = 1 +\text{sgn}(s) \sum_{n \geq 1} \gamma_n(x)L^n,
  \]we have
  \[
 k \gamma_k(x) = \gamma_1(x)\left(\text{sgn}(s) + |s|x\frac{d}{dx}\right)\gamma_{k-1}(x)
  \]
\end{prop}

The next step of \cite{Ymem, kythesis} was to observe that there exists unique $r_k, r_{k,i} \in \mathbb{R}$ for $k\geq 1$, $1\leq i < k$ such that
\begin{align*}
  & \sum_{k \geq 1}x^kG\left(x,\frac{d}{d(-\rho)}\right)^{1+sk}(e^{-L\rho}-1)F_{k}(\rho) \bigg|_{\rho=0} \\
  & = \sum_{k \geq 1}x^kG\left(x,\frac{d}{d(-\rho)}\right)^{1+sk}(e^{-L\rho}-1)\left(\frac{r_k}{\rho(1-\rho)} + \sum_{1\leq i < k}\frac{r_{k,i}L^i}{\rho}\right) \bigg|_{\rho=0}
\end{align*}
and then use the $r_k$ and $r_{k,1}$ to build a series $P(x)$, which is the mysterious input to the differential equations studied in \cite{vBKUY, vBKUY2}.  The most unsatisfying thing about this entire framework for Dyson-Schwinger equations is the lack of understanding of $P(x)$. The chord diagram expansion is meant, among other things, as a contribution to the understanding of $P$ since in the case it applies to, it also gives $P$ as an explicit expansion over chord diagrams.

For a more general understanding we could first try to simply expand out and calculate the first few $r_k$ and $r_{k,i}$ in some examples.  Here problems begin for the most prosaic of reasons: an error.  Example 6.2 of \cite{Ymem} gives $r_k$ for $k\leq 5$ and $r_{k,i}$ for $k\leq 4$.  Unfortunately it is wrong.  The following example corrects this error.

\begin{example}
This example corrects Example 6.2 of \cite{Ymem}.  Note that the conventions of this paper differ from those of \cite{Ymem} in two ways, first the indexing of the $f_i$ has been shifted by $1$ to make the chord diagram construction simpler, and second the sign of $s$ has been swapped.

Consider the case with $s=-2$ and a single $B_+$ at order $k=1$.  This is the case we now fully understand in terms of chord diagrams.  Write 
\[
F(\rho) = \sum_{j=-1}^{\infty} f_{j+1}\rho^j
\]
Then
\begin{align*}
  r_1 & = f_0 \\
  r_2 & = f_0f_1-f_0^2 \\
  r_{2,1} & = 0 \\
  r_3 & = -4f_0^2f_1+3f_2f_0^2+f_0f_1^2 \\
  r_{3,1} & = 0 \\
  r_{3,2} & = 0 \\
  r_4 & = 11f_2f_0^2f_1-9f_0^2f_1^2-18f_2f_0^3+f_0f_1^3+15f_3f_0^3 \\
  r_{4,1} & =0 \\
  r_{4,2} & =0 \\
  r_{4,3} & =0 \\
  r_{5} & = 86f_3f_0^3f_1-120f_3f_0^4-16f_0^2f_1^3+f_0f_1^4+30f_2^2f_0^3+105f_0^4f_4\\
  & \quad -112f_2f_0^3f_1+26f_2f_0^2f_1^2
\end{align*}
The problem with the original computation was a sign error in the program.  This new computation has been verified in two ways; first the example was independently computed by Erik Panzer, and second the $\gamma_i$ calculated by the program are now verified to satisfy Proposition \ref{prop gammak rec}.
\end{example}

The calculation brings up a very interesting point.  Only the $r_i$ are needed.  This is a very welcome development. The intuition of the rewriting with the $r_i$ and $r_{i,k}$ is that we are replacing the Mellin transforms of the primitives, $F_k(\rho)$, with geometric series; the sum term was simply a higher order hack to make things work out.  What this author was not able to see at the time is that the hack is simply not needed.  In fact any function of the form
\[
 g_k(\rho) = \frac{1}{\rho} + O(\rho^0)
\]
can take the place of $1/(\rho(1-\rho))$ with no need for any $r_{k,i}$\footnote{Thanks to a referee for pointing out that this is the correct level of generality for this result and for pointing out the tidier proof included here.}.  For the purposes of the program of \cite{Ymem, kythesis} the most useful $g_k(\rho)$ are $g_k(\rho) = 1/(\rho(1-\rho))$ or $g_k(\rho) = 1/(\rho(1+\rho))$ since the next steps of \cite{Ymem, kythesis} use $1 \pm \rho^2 g_k(\rho) = \rho g_k(\rho)$ to simultaneously keep the order of the ultimate differential equation low and the expression for $P$ relatively simple.
 
\begin{theorem}\label{thm geometric}
  Let $s$ be an integer.
  Let $G(x,L) = 1 +\text{sgn}(s) \sum_{n \geq 1} \gamma_n(x)L^n$ solve the Dyson-Schwinger equation \eqref{eq aDSE redux}.  Let $g_k(\rho) = \frac{1}{\rho} + O(\rho^0)$
  Then there exists unique $r_k \in \mathbb{R}$ for $k\geq 1$ such that
\begin{align*}
  & \sum_{k \geq 1}x^kG\left(x,\frac{d}{d(-\rho)}\right)^{1+sk}(e^{-L\rho}-1)F_{k}(\rho) \bigg|_{\rho=0} \\
  & = \sum_{k \geq 1}x^kG\left(x,\frac{d}{d(-\rho)}\right)^{1+sk}(e^{-L\rho}-1)r_kg_k(\rho) \bigg|_{\rho=0}
\end{align*}
  In particular, the reduction to geometric series in \cite{Ymem, kythesis} can be a pure reduction to geometric series, without fudge factors at higher powers in $L$.
\end{theorem}

\begin{proof}
First note that it suffices to prove the result for $g_k = 1/(\rho(1-\rho))$ since the full result then follows by applying this smaller result twice, once with the $F_k$ and then again with the general $g_k$.

Next note that Proposition \ref{prop gammak rec} holds independently of the values of the $f_{k,i}$, and hence in particular it holds when $F_k(\rho) = r_k/(\rho(1-\rho))$ for any $r_k\in \mathbb{R}$.  
Write $\gamma_k'$ for the $\gamma_k$ with these particular $F_k(\rho)$.

Write
\begin{align*}
P(x) & = \sum_{k\geq 1}x^kG\left(x, \frac{d}{d(-\rho)}\right)^{1+sk}(\rho - \rho^2)F_k(\rho)\bigg|_{\rho=0} \\
P'(x) & = \sum_{k\geq 1}x^kG\left(x, \frac{d}{d(-\rho)}\right)^{1+sk}(\rho - \rho^2)\frac{r_k}{\rho(1-\rho)}\bigg|_{\rho=0}
\end{align*}
Then $P'(x) = \sum_{k \geq 1}r_kx^k$ and from \eqref{eq aDSE redux}
\[
P(x) = -\gamma_1(x)-2\gamma_2(x) \quad \text{and} \quad P'(x) = -\gamma_1'(x)-2\gamma_2'(x).
\]
Choose the $r_k$ so that $P(x)=P'(x)$ as series expansions.  Then by Proposition \ref{prop gammak rec}
\begin{align*}
& -\gamma_1(x) - \gamma_1(x)\left(\text{sgn}(s) + |s|x\frac{d}{dx}\right)\gamma_1(x) \\
& = P(x) = -\gamma_1'(x) - \gamma_1'(x)\left(\text{sgn}(s) + |s|x\frac{d}{dx}\right)\gamma_1'(x)
\end{align*}
Taking this coefficient by coefficient we get
\begin{align*}
& -\gamma_{1,i} - \sum_{j=1}^{i-1}\gamma_{1,j}(\text{sgn}(s) + |s|(i-j))\gamma_{1,i-j}\\
& = r_i = -\gamma_{1,i}' - \sum_{j=1}^{i-1}\gamma_{1,j}'(\text{sgn}(s) + |s|(i-j))\gamma_{1,i-j}'
\end{align*}
Therefore $\gamma_{1,1} = r_1 = \gamma_{1,1}'$ and inductively $\gamma_{1,i}=\gamma_{1,i}'$ for all $i$.  The result follows.
\end{proof}

This theorem is conceptually important because it says that rather than performing an ad-hoc transformation what we are doing is converting our problem from the original theory with its analytically complicated primitives, to a modified theory with new primitives each of which is analytically simply a geometric series.  In this way the analytic complexity of the Dyson-Schwinger equation has been unwound into the combinatorics.

\section{Denominator reduction and special changes of variables}\label{sec cov}

This section will investigate a different use of combinatorial interpretations in the study of periods in quantum field theory.  Here we will consider the values of individual graphs, specifically primitive massless scalar $\phi^4$ graphs.

\subsection{Polynomials and denominator reduction}

For us a graph may have multiple edges and may have loops in the sense of graph theory.
For any graph $G$, $E(G)$ represents the set of edges of $G$ and $V(G)$ represents the set of vertices of $G$.  For $G$ connected $\ell(G)$ is the first Betti number of $G$, equivalently the number of independent cycles in $G$.  $G/e$ denotes $G$ with the edge $e$ contracted, and $G\setminus e$ denotes $G$ with the edge $e$ deleted.

\begin{definition}\label{def K}
Let $K$ be a connected 4-regular graph.  Let $G$ be $K$ with one vertex removed and suppose $G$ is also connected.  Then we say $G$ is a \emph{4-point graph in $\phi^4$}.  
\end{definition}

\begin{definition}
  Let $G$ be a $4$-point graph in $\phi^4$.  $G$ is \emph{primitive} if every proper subgraph $\gamma$ of $G$ with at least one edge satisfies $|E(\gamma)|>2\ell(\gamma)$
\end{definition}  
Note that this is primitivity in the renormalization Hopf algebra because the condition in $\phi^4$ for $\gamma$ to be a subdivergence is $|E(\gamma)| \leq 2\ell(\gamma)$.

\begin{definition}
  A graph $H$ is \emph{internally k-edge connected} if for every set $S$ of $k-1$ edges of $H$ either
  \begin{itemize}
    \item $H \setminus S$ is connected (so $S$ is not a cut set), or
    \item $H \setminus S$ has two components, one of which is an isolated vertex.
  \end{itemize}
\end{definition}

\begin{prop}
  Let $G$ be a $4$-point graph in $\phi^4$.
  $G$ is primitive iff $K$ is internally 6-edge connected, where $K$ is as  in Definition \ref{def K}.
\end{prop}

The idea here is simply that an internal 4-edge cut is exactly a subdivergence.
\begin{proof}
  Let $v$ be the vertex of $K$ removed to create $G$.
  Since all vertices of $K$ have even degree $K$ has no cut sets of odd size.  Thus $K$ not internally 6-edge connected is equivalent to there being an internal 4-edge cut.

  \medskip

  Suppose $K$ is not internally 6-edge connected.  Since $K$ is 4-regular, for any set $S$ of four edges of $K$, $K\setminus S$ can have at most one isolated vertex as a component.  Furthermore, each connected component of $K\setminus S$ viewed as a subgraph of $K$ must be incident to an even number of edges of $S$, as if not then the set of vertices of $K\setminus S$ would be incident to an odd number of edges in $S$ contradicting 4-regularity.   

Thus, by hypothesis, there is a set $S$ of four edges of $K$ such that $K\setminus S$ has at least two components neither of which are isolated vertices.  Call these two components $K_1$ and $K_2$.  

By Euler's formula $E(K_i) \leq 2\ell(K_i)$.  At least one of $K_1$ or $K_2$ does not include $v$, but then this is a divergent subgraph of $G$ and so $G$ is not primitive.
\medskip

Suppose $G$ is not primitive.  Then there is a proper subgraph $\gamma$ of $G$ with at least one edge and with $|E(\gamma)| \leq 2\ell(\gamma)$.  Consider $\gamma$ as a subgraph of $K$.  By 4-regularity and Euler's formula there are at most 4 edges connecting $\gamma$ to the rest of $K$.  These 4 edges form an edge cut of $K$.  It is an internal cut since $\gamma$ has at least one edge and is a proper subgraph of $G$.  Thus $K$ is not internally 6-edge connected.
\end{proof}

\medskip

With those observations out of the way, let's move to some graph polynomials.

\begin{definition}
  Let $G$ be a graph.  Take a variable $a_e$ for each edge $e$ of $G$.  Then
  \[
  \Psi_G = \prod_{T} \sum_{e\not\in T}a_e
  \]
  where the sum is over all spanning trees of $G$.  Call $\Psi_G$ the \emph{Kirchhoff polynomial} of $G$.
\end{definition}

Note this is the dual definition to what is sometimes called the Kirchhoff polynomial.

For example if 
\[
 G = \includegraphics{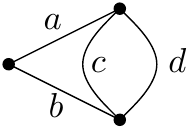}
\]
then $\Psi_G = cd + (a+b)(c+d)$.  This example is not primitive, but was chosen for its smallness.

The period of $G$, its residue as a Feynman integral, is as follows.
\begin{definition}\label{def feynman integral}
  Let $G$ be a primitive $4$-point graph in $\phi^4$.  The \emph{period} of $G$ is
  \[
  \int_{a_i\geq 0} \frac{\sum_{i=1}^{|E(G)|} (-1)^i da_1\wedge \cdots \widehat{da_i} \cdots \wedge da_{|E(G)|}}{\Psi_G^2}
  \]
\end{definition}
This integral converges since $G$ is primitive.  Really this integral should be considered as a projective integral over $a_i \geq 0$ in $\mathbb{PR}^{|E(G)|-1}$, but the above definition keeps things elementary and suffices for our needs.

By the matrix tree theorem $\Psi_G$ can also be expressed as a determinant.

\begin{prop}\label{prop matrix tree}
  Orient $G$ and let $E$ be the $|V(G)|\times |E(G)|$ incidence matrix of $G$.  Let $\widetilde{E}$ be $E$ with any one row removed and let $\Lambda$ be the diagonal matrix of the $a_e$ in the same order as the columns of $E$.  Then
  \[
  \Psi_G = \det\begin{bmatrix} \Lambda & \widetilde{E}^T \\ -\widetilde{E} & 0\end{bmatrix}
  \]
and in particular does not depend on the choices used to construct the matrix.
\end{prop}

\begin{proof}
  There are a number of ways to approach this proof (see for example Proposition 21 of \cite{Brbig}) but in the end it always comes down to the matrix-tree theorem.  The shortest proof the author knows is as follows (see \cite{VY}).  

$\Lambda$ is invertible, so we
can calculate the determinant using the Schur complement:
\[
  \det\begin{bmatrix} \Lambda & \widetilde{E}^T \\ -\widetilde{E} & 0\end{bmatrix} = a_1\cdots a_{|E(G)|} \det(0 - (-\widetilde{E}\Lambda^{-1}\widetilde{E}^T ))
= a_1\cdots a_{|E(G)|} \det(\widetilde{E}\Lambda^{-1}\widetilde{E}^T)
\]
For an $m\times n$ matrix $A$ and $S\subseteq \{1,\cdots,n\}$ let $A_S$ be the submatrix of $A$ given by columns indexed by $S$.  By the Cauchy-Binet formula
\begin{align*}
 \det(\widetilde{E}\Lambda^{-1}\widetilde{E}^T) & = \sum_{\substack{S\subseteq E(G)\\ |S|=|V(S)|-1}}\det(\widetilde{E}_S) \det((\Lambda^{-1}\widetilde{E}^T)^T_S) \\
 & = \sum_{\substack{S\subseteq E(G)\\ |S|=|V(S)|-1}} \prod_{i\in S}\frac{1}{a_i} \det(\widetilde{E}_S)^2
\end{align*}
The matrix tree theorem says that $\det\widetilde{E}_S = \pm 1$ if $S$ is a spanning tree of $G$ and is $0$ otherwise.  The result follows.
\end{proof}

A key part of Brown's approach \cite{Brbig} to calculating the periods of these graphs is denominator reduction.  First we need some polynomials built from the matrix of the previous proposition.

\begin{definition}
  Let $G$, $\Lambda$ and $\widetilde{E}$ be as in Proposition \ref{prop matrix tree}.  Let 
  \[
  M = \begin{bmatrix} \Lambda & \widetilde{E}^T \\ -\widetilde{E} & 0\end{bmatrix}.
  \]
  Let $I$ and $J$ and $K$ be sets of edge indices. Let $M(I,J)$ be the matrix obtained from $M$ by removing the rows indexed by $I$ and the columns indexed by $J$.  Suppose $|I|=|J|$.  Then the polynomial
  \[
    \Psi^{I,J}_{G,K} = \det M(I,J) |_{\substack{a_i=0\\i\in K}}
  \]
  is called a \emph{Dodgson polynomial}.
\end{definition}

Dodgson polynomials satisfy a contraction-deletion relation.

\begin{prop}\label{prop contr del}
  For $\ell\not\in I\cup J\cup K$
  \[
  \Psi^{I,J}_{G,K} = \Psi^{I\ell,K\ell}_{G,K}a_\ell + \Psi^{I,J}_{G,K\ell} 
  \]
  and
  \[
  \Psi^{I\ell,J\ell}_{G,K} = \Psi^{I,J}_{G\setminus \ell,K} \quad  \Psi^{I,J}_{G,K\ell} = \Psi^{I,J}_{G/\ell,K}
  \]
\end{prop}

\begin{proof}
  These follow from the form of the matrix defining the Dodgson polynomials.  (See for example \cite{BrS} subsection 2.2.)
\end{proof}

\begin{prop}
  Given 5 distinct edge indices $i,j,k,l,m$
  \[
  {}^5\Psi_G(i,j,k,l,m) = \pm(\Psi^{ij,kl}_{G,m}\Psi^{ikm,jlm}_G - \Psi^{ik,jl}_{G,m}\Psi^{ijm,klm}_G)
  \]
  is independent (up to overall sign) on the order of $i,j,k,l,m$.
\end{prop}

For a proof see Lemma 87 in \cite{Brbig}.

\begin{definition}
${}^5\Psi_G(i,j,k,l,m)$ is called the \emph{5-invariant} of $G$ depending on edges $i,j,k,l,m$.
\end{definition}

\begin{definition}
  Given $G$ with at least $5$ edges, \emph{denominator reduction} is a sequence of polynomials $D^5$, $D^6$, \ldots $D^k$, defined by
  \begin{itemize}
    \item $D^5_G(i_1,i_2,i_3,i_4,i_5) = {}^5\Psi_G(i_1,i_2,i_3,i_4,i_5)$
    \item If $D^j_G(i_1,\ldots,i_j)$ can be factored as 
      \[
       D^j_G(i_1,\ldots, i_j) = (Aa_\ell + B)(Ca_\ell + D)
      \]
      where $A,B,C,D$ are polynomials (not necessarily nonzero) in the edge variables not involving $a_\ell$
      then $D^{j+1}_G(i_1,\ldots,i_j,\ell) = \pm(AD-BC)$
    \item If $D^{j+1}_G=0$ or $D^j_G$ cannot be factored then denominator reduction ends.
  \end{itemize}
\end{definition}

There are a number of things to observe about this definition.  First, $D^j_G(i_1,\ldots,i_j)$ is defined up to overall sign.

Second, for a given graph, different edge orders will give a different sequence of polynomials, and such sequences may not all be the same length.  The goal is to find long sequences because of the next observation.

Third, the sequences of polynomials given by denominator reduction calculate at least two useful things.  Namely they give the denominators of the result of taking the Feynman integral of $G$ and integrating the indicated variables \cite{Brbig}, and furthermore they carry a significant part of the information of the integral which can be seen from the fact that the $c_2$ invariant of $G$ \cite{BrS}, an arithmetic invariant, can be calculated from the denominators.

Fourth, we can begin the sequence with $D_4$ rather than $D_5$ at the expense of the fact that for any choice of four edges there are three generically distinct possible choices for $D_4$, 
\[
\Psi^{ij,kl}\Psi^{ik,jl}, \Psi^{ij,kl}\Psi^{il,jk}, \text{ or } \Psi^{ik,jl}\Psi^{il,jk}
\]
each of which yields the same $D_5$ and onwards following the denominator reduction algorithm.  The denominators of the Feynman integral with fewer than 4 edges integrated also have nice explicit forms, see \cite{Brbig}, but they don't relate by the same identity.

Denominator reduction is appealing in the context of this paper because the resulting polynomials usually have combinatorial interpretations as sums of products of spanning forest polynomials \cite{BrY}.  The true strength of denominator reduction is in the interplay of these many viewpoints.

\medskip

To advance with these observations we need to define the $c_2$ invariant and spanning forest polynomials.

\begin{definition}
  Let $f_1,\ldots f_k$ be polynomials in the variables $x_1, \ldots, x_n$.  Let $q$ be a prime power.  Let $V(f_1,\ldots, f_k)$ be the affine variety defined by the vanishing of the $f_i$.  Define
  \[
  [f_1,\ldots,f_k]_q
  \]
  to be the number of points in $V(f_1, \ldots, f_k)$ over $\mathbb{F}_q$, the finite field with $q$ elements.
\end{definition}

Brown and Schnetz showed \cite[Theorems 2, 3, and 29]{BrS}
\begin{prop}
  Let $G$ be any connected graph with at least 3 vertices, then
  \[
    [\Psi_G]_q \equiv c_2(G)_q \,q^2 \mod q^3
  \]
  for some $c_2(G)_q\in \mathbb{Z}/q\mathbb{Z}$.

  If further $G$ has at least $5$ edges, then
  \[
  c_2(G)_q \equiv (-1)^n[D^n_G(i_1,\ldots,i_n)]_q  \mod q
  \]
  for any edges $i_1,\ldots,i_n$ of $G$ for which denominator reduction is defined.
\end{prop}

\begin{definition}
  Given a connected graph $G$ with at least 3 vertices, 
  the sequence indexed by prime numbers
  \[
    c_2(G) = (c_2(G)_2, c_2(G)_3, c_2(G)_5, \ldots)
  \]
  is the $c_2$ invariant of $G$.
\end{definition}

In fact this is but one manifestation of the $c_2$ invariant, and not the deepest one \cite{BrSY}, but it suffices for the purposes of this paper.  The $c_2$ invariant is relatively easy to calculate but contains a lot of information about the period.  For example if $c_2(G)=0$ then the period of the graph has a drop in transcendental weight.

\medskip

Spanning forest polynomials give a combinatorial interpretation for Dodgson polynomials and are very handy for reasoning about them.

\begin{definition}
  Let $G$ be a graph.
  Let $P$ be a partition of a subset of the vertices of $G$.  Let $\mathcal{F}_{P,G}$ be the set of all subgraphs $F$ of $G$ with the properties that 
  \begin{itemize}
    \item $V(F) = V(G)$
    \item $F$ is acyclic
    \item The number of connected components of $F$ is the number of parts of $P$
    \item For each part $p$ of $P$, there is a connected component of $F$ which contains  all the vertices of $p$ and none of the other vertices of $P$.
  \end{itemize}
  Call the elements of $\mathcal{F}_{P,G}$ the \emph{spanning forests of $F$ consistent} with $P$.   
\end{definition}

Note that the spanning forests may have isolated vertices, but only if those vertices are alone in a part of $P$.

\begin{definition}
  Let $G$ be a graph.  Take a variable $a_e$ for each edge $e$ of $G$.  Let $P$ be a partition of a subset of the vertices of $G$.  Then
  \[
  \Phi^P_G = \sum_{F\in \mathcal{F}_{P,G}}\prod_{e\not\in F}a_e 
  \]
  Such polynomials are called \emph{spanning forest polynomials}.
\end{definition}

The number of parts of $P$ will be called the number of \emph{colours} of $\Phi^P_G$; the idea is that we could draw the graph and colour those vertices of the partition according to which part they belonged to.

If $G$ has an isolated vertex $v$, then for $\Phi^P_G$ to be nonzero, $v$ must be alone in a part of $P$.  Then as polynomials
\[
\Phi^P_G = \Phi^{P-\{v\}}_{G-v}
\]
but the former has one more colour than the latter.

The relationship between Dodgson polynomials and spanning forest polynomials is given in Proposition 12 from \cite{BrY}.
\begin{prop}\label{Psi to Phi prop}
Let $I,J,K$ be sets of edge indices of a connected graph $G$ with $|I|=|J|$, then
\begin{equation}\label{Psi to Phi} 
  \Psi^{I,J}_{G,K} = \sum_P \pm \Phi^P_{G \backslash (I\cup J\cup K)}
\end{equation}
where the sum runs over all set partitions $P$ of the end points of edges of $(I \cup J \cup K) \setminus (I\cap J)$ with the property that all the forests of $P$ become trees in both \[G \backslash I / (J\cup K) \quad \text{and} \quad G \backslash J / (I\cup K)\] 
\end{prop}

Thus Dodgson polynomials are always signed sums over spanning forest polynomials.  As we denominator reduce, the factorizations are not required to be expressible in terms of Dodgson polynomials, however in practice this is what occurs.  Then the $D^n_G$ are typically signed sums of products of pairs of spanning forest polynomials.

\subsection{Free factorizations}

Nothing that has been described so far explains how many steps of denominator reduction will be possible in a given graph.  Lets think of this in terms of free factorizations.  If we make good choices of edges to reduce then the denominator may automatically be factored for combinatorial reasons.  Call this a \emph{free} factorization.

One important source of free factorizations, which would be known to anyone working with such Feynman integrals, and was first described in this language in \cite[Section 7.4]{Brbig}, is the following proposition.

\begin{prop}\label{prop triangle and 3}
Let $I$, $J$, and $K$ be sets of edge indices of a graph $G$.
\begin{enumerate}
\item Let $\{i,j,k\}$ be a triangle in $G$.  Then if $\{i,j,k\}\subseteq (K\cup I)\setminus J$ 
  \[
  \Psi^{I,J}_{G,K}=0
  \]
  While if $\{i,j\}\subseteq (K\cup I)\setminus J$ with $k\not\in I\cup J\cup K$, then $\Psi^{I,J}_{G,K}$ is divisible by $a_k$.

\item Let $\{i,j,k\}$ be a cut set in $G$ (that is $G\setminus \{i,j,k\}$ is disconnected).  Then if $\{i,j,k\} \subseteq I$
  \[
  \Psi^{I,J}_{G,K}=0
  \]
  While if $\{i,j\}\subseteq I$ with $k\not\in I\cup J\cup K$, then $\Psi^{I,J}_{G,K}$ is independent of $a_k$.
\end{enumerate}
\end{prop}

Since we are interested in $G$ primitive, the only 3-edge cuts will be 3-valent vertices, so the second point is only interesting in the case where $i,j,k$ meet at a common 3-valent vertex.

\begin{proof}
  By Proposition \ref{prop contr del}
  \[
    \Psi^{I,J}_{G,K} = \Psi^{I\ell,J\ell}_{G,K}a_\ell + \Psi^{I,J}_{G,K\ell}.
  \] 
  for any edge index $\ell \not\in I\cup J\cup J$.
  Hence in both (1) and (2) the first statement implies the second.

  If $\{i,j,k\}$ is a cut set and $\{i,j,k\} \subseteq I$, then $G \backslash I / (J\cup K)$ is disconnected  and hence $\Psi^{I,J}_{G,K}=0$ by Proposition \ref{Psi to Phi prop}.  (See \cite{Brbig} for a proof independent of spanning forest polynomials.)

  Dually, if $\{i,j,k\}$ is a triangle and $\{i,j,k\}\subseteq (K\cup I)\setminus J$  then no spanning tree of $G \backslash J$ contains all three of the edges,  Hence setting all three variables to $0$, equivalently contracting them, gives $\Psi_{G \backslash J / (I\cup K)}=0$ and hence $\Psi^{I,J}_{G,K}=0$.
\end{proof}

The previous proposition gives us free factorizations because if we reduce two edges of the triangle or the 3-valent vertex, say in the initial 5-invariant, then by choosing an appropriate order we can make sure that in each term one of the factors either has  no constant term (triangle case) or no linear term (cut set case) in the other edge.  Thus the entire 5-invariant either has no constant term (triangle case) or has no quadratic term (cut set case).  In either case we have a trivial sort of factorization and can continue denominator reducing.

\medskip

Another important source of free factorizations is graphs which are very narrow in the following sense \cite[section 1.4]{Brbig}.

\begin{definition}
  Given a total order $<$ on the edges of a graph $G$, for $0\leq i \leq |E(G)|$, let $L_i(G)$ be the first $i$ edges in the order and $R_i(G)$ be the remaining edges in the order.  Let
  \[
  V_i(G,<) = V(L_i(G))\cap V(R_i(G))
  \]
  Then the \emph{vertex width} of $G$ is
  \[
  \min_< \left(\max_{0\leq i \leq |E(G)|} |V_i(G,<)|\right)
  \]
  where the minimum runs over all total orders of edges of the graph.
\end{definition}

This notion of vertex width is related to, but not equivalent to, the notion of path width from graph theory.

\begin{definition}
  Let $G$ be a graph.  A \emph{path decomposition} of $G$ is a sequence of sets of vertices of $G$, 
  \[
  V_0, V_1, \ldots, V_k
  \]
  such that 
  \begin{itemize}
    \item for every edge $e$ of $G$ there is a $V_i$ such that both ends of $e$ are in $V_i$ and 
    \item for every vertex $v$ of $G$, the $V_i$ in which $v$ appear form a contiguous subsequence of the original sequence.
  \end{itemize}
  The \emph{width} of the path decomposition is $\max_{0\leq i\leq k} |V_i|-1$.
\end{definition}

\begin{definition}
  The \emph{path width} of $G$ is the minimum width of a path decomposition of $G$.
\end{definition}

Path width was first defined by Robertson and Seymour \cite{RS1} and is an important notion in graph theory.  It is a special case of the tree width of a graph.

\begin{prop}
  Let $G$ be a graph, let $vw(G)$ be the vertex width of $G$ and $pw(G)$ be the path width of $G$.  Then
  \[
    vw(G) \geq pw(G)
  \]
\end{prop}

\begin{proof}
  First note that both the path width and the vertex width of $G$ is the maximal path width, respectively vertex width, of a connected component of $G$, and so we may assume $G$ is connected.
  Next note that loops of $G$ affect neither the path width nor the vertex width, so we may assume $G$ has no loops in the sense of graph theory.

  If $W$ is a set of vertices of $G$ use the notation $e\in W$ to say that both ends of $e$ are in $W$.

\medskip

  Suppose $G$ has vertex width $w$ and that $e_1<e_2<\ldots< e_{|E(G)|}$ is a total order on the edges of $G$ which gives width $w$.  If the graph has only one edge then the result holds trivially.  Otherwise let $E_i$ be the set of vertices which are ends of $e_i$.

Define
  \begin{align*}
    V_{2i-1} &= V_{i-1}(G,<)\cup E_i \quad 1\leq i \leq |E(G)| \\
    V_{2i} &= V_i(G,<) \quad 1 \leq i \leq |E(G)|-1
  \end{align*}
  We now want to check that the $V_j$ give a path decomposition of $G$.  By construction $e_i \in V_{2i-1}$.  By hypothesis the $V_j(G,<)$ all have size at most $w$.  

If $e_i \in V_{i-1}(G,<)$ then $V_{2i-1}=V_{i-1}(G,<)$ which has size at most $w$.  If one end of $e_i$ is in $V_{i-1}(G,<)$ then $V_{2i-1}$ has size at most $w+1$.

Suppose neither end of $e_i$ is in $V_{i-1}(G,<)$.  Then each vertex of $V_{i-1}(G,<)$ must meet at least one edge $e_j$ with $j>i$.  Thus $V_{i-1}(G,<)\subseteq V_i(G,<)$.  Also any edges which meet an end of $e_i$ must have index greater than $i$.  Finally by connectivity $e_i$ has at most one end of degree 1.   Thus $|E_i \cap V_i(G,<)| \geq 1$, and so $V_{2i-1}$ has size at most $w+1$.

Consider a vertex $v$ of $G$.  $v$ first appears in the $V_i(G,<)$ when the first edge incident to $v$ passes into $L_i(G)$ and last appears when the last edge incident to $v$ passes into $L_i(G)$.  Thus $v$ appears in a contiguous subset of the $V_i(G,<)$.  If $v\in E_i$ then $v\in V_{i-1}(G,<)$ or $v\in V_i(G,<)$.  Therefore $v$ appears in a contiguous subset of the $V_j$ and so the $V_j$ give a path decomposition of $G$ with parts of size at most $w+1$.  Hence $G$ has path width at most $w$.  
\end{proof}

We can see that the path width and vertex width are not equal in general by considering $K_{3,3}$

\begin{example}
  Consider $K_{3,3}$ with the vertices labelled as follows
  \[
  \includegraphics{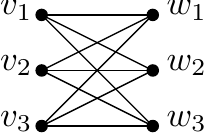}
  \]
  $K_{3,3}$ has path width at most 3, from the following path decomposition
  \[
  \{v_1,w_1,w_2,w_3\}, \{v_2,w_1,w_2,w_3\}, \{v_3,w_1,w_2,w_3\}
  \]
 
  Suppose $K_{3,3}$ had vertex width $3$.  Due to the automorphisms of $K_{3,3}$ it doesn't matter which edge we take first in the order, say we first take $\{v_1,w_1\}$. Then 
  \[
  V_1(G,<) = \{v_1,w_1\}
  \]
  To force $V_2(G,<)$ to have size at most 3 we must pick as our next edge an edge incident to $v_1$ or to $w_1$.  Without loss of generality take $\{v_1,w_2\}$ as the next edge.  Then
  \[
  V_2(G,<) = \{v_1,w_1,w_2\}
  \]
  $w_1$ and $w_2$ both have two further edges incident to them.  But we must lose a vertex from $V_2(G,<)$ for every vertex added, so as our next edge we must take $\{v_1,w_3\}$.  Then
  \[
  V_3(G,<) = \{w_1,w_2,w_3\}
  \]
  Now no matter which edge we take next we will still have $w_1$,$w_2$, and $w_3$ in $V_4(G,<)$ along with either $v_1$ or $v_2$, and hence $K_{3,3}$ is not vertex width $3$.
\end{example}

Brown in \cite{Brbig} proved that for every primitive 4-point graph in $\phi^4$ with vertex width at most 3 there is an edge order so that denominator reduction either terminates with a $0$ or continues until no edges remain; the graph is sufficiently narrow that every step gives a free factorization.  

Iain Crump in \cite{Cmsc} showed that every 3-connected graph for which all 5-invariants factor due to one of the defining polynomials being zero has a particular structure which implies vertex width at most 3.

\medskip

Triangles and 3-valent vertices are shapes in the graph which give free factorizations. There is one more shape which the author is aware of which yields free factorizations.

\begin{prop}\label{prop new free}
  Let $G$ be a graph with two 3-valent vertices $v_1$ and $v_2$ joined by an edge.  Then after reducing the five edges incident to $v_1$ and $v_2$ and any other two edges of $G$, the denominator $D^7_G$ factors into two factors each at most linear in each edge variable, and expressible in terms of Dodgson polynomials.
\end{prop}

If $G$ is a $4$-point graph in $\phi^4$ and $K$ is as in Definition \ref{def K}, then $G$ satisfying the hypotheses of the proposition implies that $K$ contains a triangle.  Thus, in view of Schnetz' completion results \cite{Sphi4}, is it not surprising that this shape yields a free factorization, however it is not a consequence of this since the missing edges from $K$ are not integrated.

Note also that a special case of Proposition \ref{prop new free} is Lemma 55 of \cite{BrS}.

\begin{proof}
  Label the edges of $G$ incident to $v_1$ and $v_2$ as illustrated
  \[
  \includegraphics{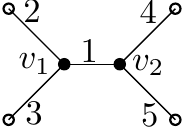}
  \]
  where the rest of the graph attaches at the white vertices.
  Let $i$ and $j$ be any other two edges of $G$.  Then
  \[
  {}^5\Psi_G(2,3,5,i,j) = \Psi^{23,ij}_{G,5}\Psi^{25i,35j}_{G} - \Psi^{235,ij5}_G\Psi^{2i,3j}_{G,5}
  \]
  By Proposition \ref{prop triangle and 3}, edge $1$ cannot be cut in the first factor of either term, so
  \[
  D^6_G(1,2,3,5,i,j) = \Psi^{23,ij}_{G,15}\Psi^{125i,135j}_{G} - \Psi^{235,ij5}_{G,1}\Psi^{12i,13j}_{G,5}
  \]
  Again by Proposition \ref{prop triangle and 3}, edge $4$ cannot be cut in $\Psi^{125i,135j}_{G}$.  Furthermore, edge $4$ cannot be cut in $\Psi^{235,ij5}_{G,1}$ because when $2,3$ and $5$ are cut, then cutting $4$ would disconnect the graph giving $0$ in the same way as in the proof of Proposition \ref{prop triangle and 3}.  Therefore
  \[
  D^7_G(1,2,3,4,5,i,j) = \Psi^{234,4ij}_{G,15}\Psi^{125i,135j}_{G,4} - \Psi^{235,ij5}_{G,14}\Psi^{124i,134j}_{G,5}
  \]
  Next observe that
  \[
  \Psi^{125i,135j}_{G,4} = \Psi^{124i,134j}_{G,5}
  \]
  since after cutting $1$, $v_2$ is a two valent vertex, and hence it makes no difference whether we cut $4$ and contract $5$ or cut $5$ and contract $4$.  In fact both these Dodgson polynomials are equal to 
  \[
  \Psi^{i,j}_{H}
  \]
  where $H$ is $G$ with vertices $v_1$ and $v_2$ and their incident edges removed. 
  Thus
  \[
  D^7_G(1,2,3,4,5,i,j) = (\Psi^{234,4ij}_{G,15}- \Psi^{235,ij5}_{G,14})\Psi^{i,j}_H
  \]
\end{proof}

It is the experience of the author that the sources of free factorizations from this subsection explain all the free factorizations which occur before the endgame of denominator reduction.  As one example, lets look at a very important graph from \cite{BrS}.

\begin{example}\label{eg GBS}
In \cite{BrS} Brown and Schnetz consider the graph
\[
G_{BS} = \raisebox{-2cm}{\includegraphics{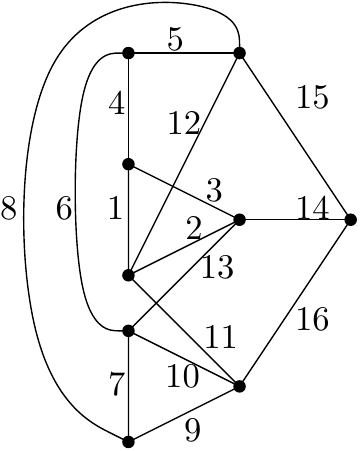}}
\]
They denominator reduce edges $1,\ldots, 10$ to get
\begin{equation}\label{eq D10}
  D^{10}_{G_{BS}}(1,\ldots, 10) = (Aa_{11}+B)(Ca_{11}+D)
\end{equation}
where
\begin{align*}
  A & = Q + a_{12}a_{13} + a_{16}a_{12} + a_{14}a_{12} + a_{15}a_{13} + a_{14}a_{13} \\
  B & = a_{13}(Q+a_{16}a_{12} + a_{14}a_{12}) \\
  C & = -a_{13}a_{15} \\
  D & = a_{12}(Q+a_{13}a_{16}) \\
  Q & = a_{14}a_{15}+a_{15}a_{16}+ a_{14}a_{16}
\end{align*}
Then $D^{11}_{G_{BS}}(1,\ldots, 11)$ does not factor, so they proceed one step further by a clever trick which is the subject of the next subsection.

Now let's see why these first 11 denominator reduction steps work out in view of the results of this subsection.
For each of the triangles with a 3-valent vertex as marked in full edges in Figure \ref{fig BSreds} we can take any two of the edges among our initial 5 edges and get the other 2 from triangles and 3-valent vertices.  This takes care of 8 of the initial integrations.  Taking either $5$ or $6$ as the fifth of the initial 5 edges we also get the other from the 3-valent vertex.  We get one free factorization because of the shape made by edges $1,3,4,5,6$, so $D_{10}(1,\ldots, 10)$ factors and we may choose any of the remaining edges to proceed with.
\begin{figure}
\includegraphics{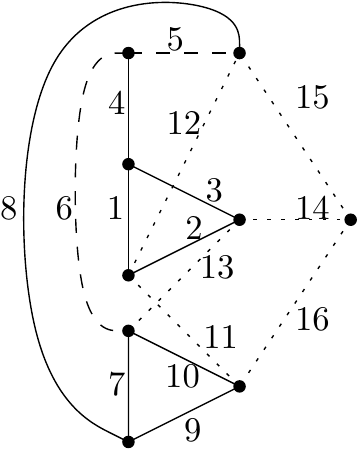}
\caption{The original Brown Schnetz K3 graph marked to illustrate the initial reductions}\label{fig BSreds}
\end{figure}\end{example}

\begin{example}
  Another interesting example is 
  \[
  \includegraphics{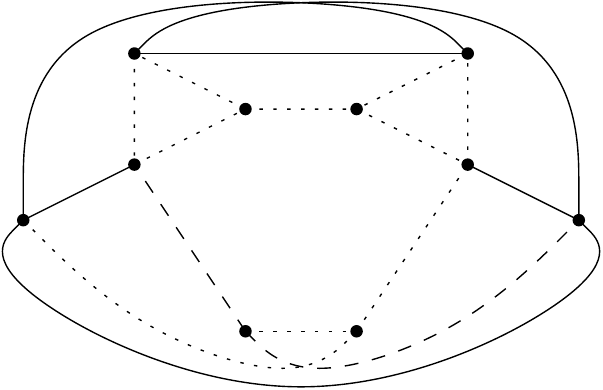}
  \]
  This graph is a decompletion of $P_{9,172}$ which appears in version 1 of \cite{Sphi4}.  Using triangles and 3-valent vertices we can choose an initial 5-invariant so as to denominator reduce all of the dotted edges.  Note that we have, in the process, reduced two 3-valent vertices which are joined by an edge.  

As our free reduction, reduce one of the dashed edges.  This also gives the other dashed edge for free on account of the 3-valent vertex.  Once again we have reduced two 3-valent vertices which are joined by an edge.  Reducing as described we obtain the two factors
\begin{align*}
&zvu+wxu+yzx+ywv+wzu+yzu+yxu+zvx\\
  & +ywz+wvx+yvu+wvu+zxu+yvx+ywx+wzv
\end{align*}
and
\[
zvu+wxu+wzu+yzu+ywz+wvu+zxu+wzv
\]
The monomials of the second are contained in the first, but if we reduce $(Aa+B)((A+C)a+(B+D))$ we get $A(B+D) - B(A+C) = AD-BC$ which is the same as reducing $(Aa+B)(Ca+D)$ so it suffices to consider only the two factors
\begin{align*}
&yzx+ywv+yxu+zvx+wvx+yvu+yvx+ywx\\
&zvu+wxu+wzu+yzu+ywz+wvu+zxu+wzv
\end{align*}
One more reduction leaves us with 5 variables.
\end{example}

\subsection{The Brown-Schnetz change of variables}

Continuing Example \ref{eg GBS} following the analysis of Brown and Schnetz of $G_{BS}$ from \cite{BrS}, next they make the change of variables
\[
  a_{12} \mapsto a_{12}Q \quad a_{13} \mapsto a_{13}Q
\]
Let $\widetilde{D}^{11}$ be $D^{11}_{G_{BS}}$ after this change of variables.  Then
\begin{align*}
\widetilde{D}^{11} & = Q^3((1+Qa_{12}a_{13}+a_{16}a_{12}+a_{14}a_{12}+a_{15}a_{13}+a_{14}a_{13})a_{12}(1+a_{13}a_{16}) \\
& \quad \quad + a_{13}^2a_{15}(1+a_{16}a_{12}+a_{14}a_{12}))
\end{align*}
Let $R$ be the polynomial $\widetilde{D}^{11}/Q^3$.
Viewing this back in the partially integrated Feynman integral, we have $\widetilde{D}^{11}$ in the denominator and some polylogarithms in the numerator along with a factor of $Q^2$ from the change of variables.  Cancelling the $Q^2$ the denominator becomes $QR$.

$Q$ is linear in all its variables and $R$ is linear in $a_{14}$ and $a_{15}$ so we can hope to continue denominator reducing.  At this stage in the argument we can't be certain whether or not we can continue -- it will depend on the numerator.  Rather than look into the form of the numerator, let's return to the approach of \cite{BrS} by considering the $c_2$ invariant. 
The question, then, is whether or not $QR$ has the same point counts modulo $p$ as $D^{11}_{G_{BS}}$.  The answer is that it does not, but that they differ only by a constant modulo $p$.  This requires the following two facts

First before the change of variables Brown and Schnetz take the opportunity to dehomogenize by setting $a_{16}=1$, so they need (see \cite{BrS} Lemma 58 for slightly less)
\begin{lemma}\label{lem dehom}
$\left[\left.D^{11}_{G_{BS}}\right|_{a_{16}=0}\right]_q$ modulo q is constant as a function of $q$.
\end{lemma}

Then as the change of variables is an isomorphism off of $V(Q)$ we have that
\[
[D^{11}_{G_{BS}}]_q - [Q,D^{11}_{G_{BS}}]_q = [R]_q - [Q,R]_q
\]
Two of these are relatively easy to control 
\begin{lemma}[\cite{BrS} Lemma 59]\label{lem other lem}
  $[Q,D^{11}_{G_{BS}}]_q$ and $[R]_q$  modulo q are constant as functions of $q$.
\end{lemma}

Brown and Schnetz then proceed to use this to show that this particular graph is a counterexample to Kontesevich's conjecture that all such graphs should have $c_2$ invariants which are polynomial in $q$.  Furthermore, taking the result of denominator reducing $QR|_{a_{16}=1}$ and changing variables one more time they are able to recognize this as defining a particular K3 surface.

\medskip

Now let's see if we can get some insight into this change of variables and why it worked.  The key is that there is a 3-valent vertex of the graph which has not been touched by the reductions so far -- a spare 3-valent vertex.  

\begin{figure}
\includegraphics{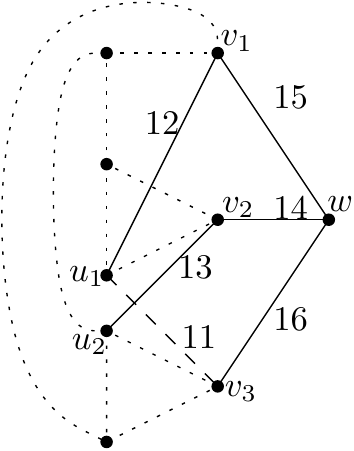}
\caption{The original Brown Schnetz K3 graph marked to illustrate the change of variables}\label{fig BScov}
\end{figure}

Figure \ref{fig BScov} illustrates what remains of the graph.  Using the vertex labels from Figure \ref{fig BScov}, note that $Q$ is the Kirchhoff polynomial of the right hand side of the graph, edges $14$, $15$, and $16$, with vertices $v_1$, $v_2$, and $v_3$ identified.  Equivalently $Q$ is the spanning forest polynomial of the same subgraph defined by the partition $\{v_1\},\{v_2\},\{v_3\}$.

The change of variables consists of scaling each of the other edges (those remaining on the left hand side of Figure \ref{fig BScov}) by $Q$. 

Each denominator is a sum of products of pairs of Dodgson polynomials and so is a sum of products of pairs of spanning forest polynomials.  
\begin{definition}
  Given a product $\Phi^{P_1}\Phi^{P_2}\cdots \Phi^{P_k}$ of spanning forest polynomials, say that the \emph{number of colours} of the product is $|P_1|+|P_2|+\cdots+|P_k|$ where $|P_i|$ is the number of parts of the partition $P_i$.
\end{definition}
By homogeneity all terms of the denominator have the same number of colours, so we can speak of the number of colours of the denominator.
By this definition $\Psi^2$ has two colours. 

Observe that a spanning forest polynomial $\Phi^P$ of a connected graph $G$ has degree $\ell(G) + |P| - 1$ since each cycle must be cut to get a tree and then the component trees must be detached to obtain a forest.

\begin{prop}
Let $G$ be a primitive $\phi^4$ graph.

If $D^j_G$ has $c$ colours then $D^{j+1}_G$ has $c+1$ colours when viewed on the same underlying graph.

If $G - \{e_{i_1},\ldots,e_{i_k}\}$ has an isolated vertex $v$ and $D^k_G(i_1,\ldots, i_k)\neq 0$ has $c$ colours, then viewing the spanning forest polynomials defined with respect to the graph with $v$ removed, $D^k_{G\setminus v}(i_1,\ldots,i_k)$ has $c-2$ colours.
\end{prop}

\begin{proof}
  Each step of denominator reduction decreases the degree by 1, so the corresponding spanning forests each have one fewer edge and hence one more component tree.  This proves the first statement.

  Since $v$ is isolated in $G-\{e_{i_1},\ldots,e_{i_k}\}$, $v$ must have its own colour in each spanning forest polynomial contributing to $D^k_G$.  Thus removing $v$ and this colour in each spanning spanning forest polynomial gives the same polynomial $D^k_{G\setminus v}$ but with two fewer colours. 
\end{proof}

Returning to the Brown Schnetz calculation,
${}^5\Psi$ has 7 colours when viewed on all the vertices of the graph.
In integrating the first 11 edges we isolate 3 vertices and so, by the previous proposition, the number of colours of $D^{11}_{G_{BS}}$, with isolated vertices removed, is $7+(11-5)-2\cdot 3 = 7$.  

Recall that $Q$ is the spanning forest polynomial defined by $\{v_1\},\{v_2\},\{v_3\}$ on the subgraph consisting of edges $14$, $15$, and $16$
 and the change of variables consists of scaling each of the other remaining edges by $Q$.

Let $\Phi$ be any spanning forest polynomial involved in $D^{11}_{G_{BS}}$.  Let $c$ be the number of colours of $\Phi$ when viewed on the graph with isolated vertices removed.  $w$ has not been involved in the denominator reduction so $w$ is not in the partition of $\Phi$.  Thus $c\leq 5$.   

If $c=5$ then $e_{12}$ and $e_{13}$ both must be cut which contributes $Q^2$ after the change of variables, and $v_1$, $v_2$, and $v_3$ are different colours which contributes another $Q$, so there is a factor of $Q^3$ after the change of variables.  If $c=4$ then either $e_{12}$ and $e_{13}$ are again cut, or exactly one of them is not cut and $v_1$, $v_2$, and $v_3$ are different colours; in both cases there is a factor of $Q^2$ after the change of variables.  If $c=3$ then either at least one of $e_{12}$ and $e_{13}$ is cut or $v_1$, $v_2$, and $v_3$ are different colours; in both cases there is a factor of $Q$ after the change of variables.  

Thus all pairs with 7 colours total lead to a factor of $Q^3$ after the change of variables, which explains the factorization of $\widetilde{D}^{11}$.

The general result to which the above calculation is a special case would be

\begin{theorem}\label{thm cov}
Let $G$ be a primitive 4-point graph in $\phi^4$.  Suppose the edges of $G$ can be partitioned into 3 parts, $G_1$, $G_2$, and $G_3$.  Suppose
\begin{itemize}
  \item $G_2\cup G_3$ is connected.
  \item $G_1\cup G_2$ and $G_3$ share $n$ vertices; call them $v_1,\ldots, v_n$.
  \item There are $v$ vertices involved in $G_1$ but not $G_2\cup G_3$.
  \item The edges of $G_1$ can be denominator reduced in the graph $G$, and the resulting denominator $D$ can be written as a sum of products of pairs of spanning forest polynomials of $G$ where the parts do not involve any vertices of $G_3$ other than possibly $v_1, \ldots, v_n$.
  \item 
    $
    2\ell(G_2\cup G_3)-2\ell(G_3)+|G_1|-|G_2|-2v-2n+3\geq 0
    $
\end{itemize}
Let $Q$ be the spanning forest polynomial for $G_3$ for the partition $\{v_1\},\ldots,\{v_n\}$.  Let $\widetilde{D}$ be $D$ with the variables of $G_2$ scaled by $Q$.  Then $Q^{|G_2|+1}$ divides $\widetilde{D}$.
\end{theorem}

Note that in the example of $G_{BS}$, $G_1$ consists of edges $1$ through $11$, $G_2$ consists of edges $12$ and $13$, and $G_3$ consists of edges $14$, $15$, and $16$.  

Note also that the assumption that the spanning forest polynomials contributing to $D$ only involve $v_1, \ldots v_n$ among the vertices of $G_3$ is natural.  This is because the other vertices of $G_3$ have not yet had any incident edges integrated, and so should not appear in any Dodgson polynomials or any spanning forest polynomials which have arisen in the usual way.  However, the assumption is necessary since we cannot rule out an unexpected factorization which might yield such anomalous partitions.

\begin{proof}
  Count colours as in the example.  Any 5-invariant of $G$ will have $7$ colours by the form of a 5-invariant.  Viewed on $G$, denominator reducing the edges of $G_1$ will give a denominator with $7+(|G_1|-5) = |G_1|+2$ colours.  Discarding each of the $v$ isolated vertices leaves $|G_1|+2-2v$ colours in $D$.  

Consider a spanning forest polynomial $\Phi$ which appears in $D$.  Suppose $\Phi$ has $c$ colours viewed on the graph $G_2\cup G_3$.  Then $\Phi$ has degree $\ell(G_2\cup G_3)+c-1$.  Let $\widetilde{\Phi}$ be $\Phi$ after scaling each variable of $G_2$ by $Q$.  

Every spanning forest of $\Phi$ induces a partition on $\{v_1,\ldots, v_n\}$ in the following way: restrict the forest to $G_3$ and partition $\{v_1,\ldots,v_n\}$ by putting vertices in the same part iff they are in the same tree in the restriction. Gather together those terms of $\Phi$ corresponding to the same partition.  

Let $P$ be such a partition with $k\leq n$ parts, $\Phi'$ the terms of $\Phi$ corresponding to $P$, and $\widetilde{\Phi}'$  those terms after the substitution.  Since $P$ has $k$ parts we must cut each cycle of $G_3$ and then $k-1$ more edges of $G_3$ and no more.  So $\Phi'$ is homogeneous of degree $\ell(G_3)+k-1$ in the $G_3$ variables.  Thus $\Phi'$ is homogeneous of degree
\[
\ell(G_2\cup G_3)+c-1 - \ell(G_3)-k+1 = \ell(G_2\cup G_3)-\ell(G_3) +c-k
\]
in the $G_2$ variables.  If $k=n$ then $Q|\Phi'$ by definition of $Q$.  So if $k=n$, $Q^{ \ell(G_2\cup G_3)-\ell(G_3) +c -n+1}$ divides $\widetilde{\Phi}'$ and if $k<n$, $Q^{ \ell(G_2\cup G_3)-\ell(G_3) +c -k}$ divides $\widetilde{\Phi}'$, so in all cases at least 
\[
\ell(G_2\cup G_3)-\ell(G_3) +c -n+1
\]
powers of $Q$ divide $\widetilde{\Phi}'$ and hence divide $\widetilde{\Phi}$.  Recalling that $D$ has $|G_1|+2-2v$ colours we get that at least
\begin{align*}
& 2\ell(G_2\cup G_3)-2\ell(G_3) + |G_1|+2-2v - 2n+2 \\
&  = 2\ell(G_2\cup G_3) - 2\ell(G_3) +|G_1|-2v-2n+4 \\
& \geq |G_2|+1
\end{align*}
powers of $Q$ divide $\widetilde{D}$, where the inequality is by hypothesis.
\end{proof}

Finally consider the constant mod $q$ conditions in the Brown-Schnetz example.  Here there is not so much to say combinatorially.

Setting a variable to $0$ corresponds to contracting it in the original graph (Proposition \ref{prop contr del}).  Then $\left.D^{11}_{G_{BS}}\right|_{a_{16}=0}$ is the same as $D^{11}_{G_{BS}/a_{16}}$.  $G_{BS}/a_{16}$ is a strictly simpler graph which can be denominator reduced all the way and so the $c_2$ invariant is constant giving Lemma \ref{lem dehom}.

A key fact for $[R]_q$ is that $R$ is linear in $a_{14}$ and $a_{15}$.  This fact is also necessary for denominator reducing $QR$ one more step.  To see that the fact is true return to the decomposition of $D^{10}_{G_{BS}}$ in \eqref{eq D10}.  Of the variables $a_{14}$, $a_{15}$ and $a_{16}$, $C$ depends only on $a_{15}$, $B$ depends only on $a_{14}$ and $a_{16}$ apart from the $Q$ which gets pulled out, and $D$ depends only on $a_{16}$ apart from the $Q$ which gets pulled out.  Therefore $[R]_q$ depends only on the coefficient of $a_{14}a_{15}$ in $R$.  The remainder of Lemma \ref{lem other lem} is a computation for which this author has no particular insight.
%\[
%  a_{12}^2*a_{13}+a_{12}^2*a_{13}^2*a_{16}+a_{12}*a_{13}^2.
%\]

\medskip

Looking back at the whole calculation the key is that there is a decomposition of $G_{BS}$ into $G_1$, $G_2$, $G_3$ where $G_3$ is small and nice while still allowing $G_1$ to be large and denominator reducible.  In other words a $3$-valent vertex gives a nice $G_3$, as long as losing this vertex to $G_1$ still lets us reduce maximally far, so the key is to have an \emph{extra} $3$-valent vertex.

\medskip

Let's consider a few other examples of using Theorem \ref{thm cov}.  

\begin{example}
Begin with 
\[
\includegraphics{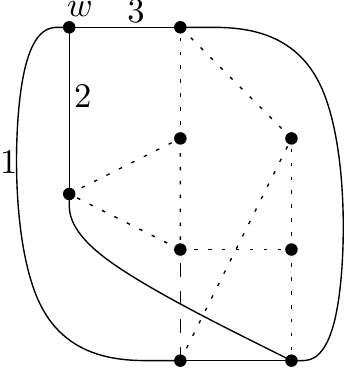}
\]
which is a decompletion of $P_{8,38}$ from \cite{Sphi4}.  Then, as in the previous subsection, we can denominator reduce the dotted edges and the result will factor.  Additionally denominator reduce the dashed edge. 

Now we are in a position to apply Theorem \ref{thm cov} with $G_1$ the dotted and dashed edges, $G_3$ the edges incident to $w$ and $G_2$ the remaining 3 edges.  This gives that $n=3$ and $v=4$ in the statement of the theorem.  Thus
\begin{align*}
& 2\ell(G_2\cup G_3) - 2\ell(G_3) + |G_1| - |G_2| - 2v - 2n + 3 \\
& = 4 - 0 + 10 - 3 - 8 - 6 + 3  \\
& = 0
\end{align*}
Then by the theorem we can make the change of variable $a_e\leftarrow a_e(a_1a_2+a_1a_3+a_2a_3)$ for $a_e\in G_2$ and the result will be divisible by $(a_1a_2+a_1a_3+a_2a_3)^4$.  Let $R$ be the result of this division.  Then it turns out $R$ is linear in $a_1$ and so we can perform one further reduction.

Note that we were able to do one fewer initial reduction compared to $G_{BS}$, so ultimately we have one more variable in play.
\end{example}

\begin{example}
Now consider the graph
\[
\includegraphics{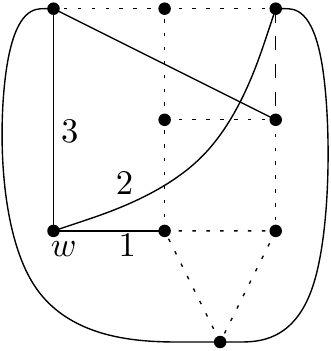}
\]
which is a decompletion of $P_{8,39}$ from \cite{Sphi4}.  Using the previous subsection we can denominator reduce the dotted edges and the result will factor.  Additionally denominator reduce the dashed edge. 

Apply Theorem \ref{thm cov} with $G_1$ the dotted and dashed edges, $G_3$ the edges incident to $w$ and $G_2$ the remaining 3 edges.  This gives that $n=3$ and $v=3$ in the statement of the theorem.  Thus
\begin{align*}
& 2\ell(G_2\cup G_3) - 2\ell(G_3) + |G_1| - |G_2| - 2v - 2n + 3 \\
& = 2 - 0 + 10 - 3 - 6 - 6 + 3  \\
& = 0
\end{align*}
Then once again by the theorem we can make the change of variable $a_e\leftarrow a_e(a_1a_2+a_1a_3+a_2a_3)$ for $a_e\in G_2$ and the result will be divisible by $(a_1a_2+a_1a_3+a_2a_3)^4$.  Let $R$ be the result of this division.  It turns out again that $R$ is linear in $a_1$ and so we can perform one further reduction.

This is not really a satisfactory answer for this graph since yet again we have one more variable remaining compared to $G_{BS}$ despite the fact that by the point counts this should not be necessary for this graph.  Such changes of variables were known to Brown and Schnetz (personal communication, 2010); the difficulty is to get one further step.  However, it is nice to know that at least this far the calculations aren't ad-hoc, but rather are directly analogous to the change of variables of \cite{BrS}.
\end{example}

\bibliographystyle{plain}
\bibliography{main}

\end{document}